\documentclass[lettersize,journal]{IEEEtran}
\usepackage{amsmath,mathrsfs,amsfonts,amssymb}
\usepackage{algorithmic,algorithm}
\usepackage{array,cuted}
\usepackage[caption=false,font=normalsize,labelfont=sf,textfont=sf]{subfig}
\usepackage{textcomp}
\usepackage{stfloats}
\usepackage{url}
\usepackage{verbatim}
\usepackage{graphicx,stfloats}
\usepackage{cite}
\usepackage{amsthm}
\newcounter{TempEqCnt}
\begin{document}

\title{Fast Ray-Tracing-Based Precise Underwater Acoustic Localization without Prior Acknowledgment of Target Depth}

\author{Wei Huang~\IEEEmembership{Member,~IEEE,} Hao Zhang~\IEEEmembership{Senior Member,~IEEE,} Kaitao Meng,~\IEEEmembership{Member,~IEEE,}, Fan Gao, Wenzhou Sun, Jianxu Shu, Tianhe Xu*, and Deshi Li*
% <-this % stops a space
\thanks{Manuscript received September 19, 2023; revised XX XX, 2023.}
\thanks{Wei Huang and Hao Zhang are with the Faculty of Information Science and Engineering, Ocean University of China, Qingdao, China (e--mail: hw@ouc.edu.cn, zhanghao@ouc.edu.cn); Kaitao Meng is with the Department of Electronic and Electrical Engineering, University College London, U.K.; Tianhe Xu, Fan Gao, and Jianxu Shu are with the Institute of Space Sciences, Shandong University, Weihai, China (e--mail:thxu@sdu.edu.cn, gaofan@sdu.edu.cn); Deshi Li is with the Electronic Information School, Wuhan University, Wuhan, China (e--mail: dsli@whu.edu.cn); Wenzhou Sun is with the State Key Laboratory of Geo-information Engineering, Xi'an, China. (Corresponding author: Tianhe Xu and Deshi Li.)}}

\markboth{IEEE Transactions on Geoscience and Remote Sensing, VOL. X, NO. X, XX, 2023}%
{Huang \MakeLowercase{\textit{et al.}}: Fast Ray-Tracing-Based Precise Underwater Acoustic Localization without Prior Acknowledgment of Target Depth}

%\IEEEpubid{0000--0000/00\$00.00~\copyright~2023 IEEE}

\maketitle

\begin{abstract}
Underwater localization is of great importance for marine observation and building positioning, navigation, timing (PNT) systems that could be widely applied in disaster warning, underwater rescues and resources exploration. The uneven distribution of underwater sound velocity poses great challenge for precise underwater positioning. The current soundline correction positioning method mainly aims at scenarios with known target depth. However, for nodes that are non-cooperative nodes or lack of depth information, soundline tracking strategies cannot work well due to nonunique positional solutions. To tackle this issue, we propose an iterative ray tracing 3D underwater localization (IRTUL) method for stratification compensation. To demonstrate the feasibility of fast stratification compensation, we first derive the signal path as a function of glancing angle, and then prove that the signal propagation time and horizontal propagation distance are monotonic functions of the initial grazing angle, so that fast ray tracing can be achieved. Then, we propose an sound velocity profile (SVP) simplification method, which reduces the computational cost of ray tracing. Experimental results show that the IRTUL has the most significant distance correction in the depth direction, and the average accuracy of IRTUL has been improved by about 3 meters compared to localization model with constant sound velocity. Also, the simplified SVP can significantly improve real-time performance with average accuracy loss less than 0.2 m when used for positioning.
\end{abstract}


\section{Introduction}
\IEEEPARstart{U}{nderwater} localization is one of the most essential techniques for building marine observation systems and positioning, navigation, timing (PNT) systems, which could be widely applied in disaster warning, underwater rescue and resources exploring. Compared with radio or optical signal, acoustic wave has become the most popular signal carrier in underwater localization \cite{Tan2011LocationSurvey}. However, unlike terrestrial radio, underwater acoustic localization faces with many challenges due to the special underwater environment, such as the difficulty in clock synchronization caused by long signal propagation delay \cite{Tan2011LocationSurvey,Qu2016LocattionSurvey}, the multipath effect caused by signal reflection at the ocean surface or bottom \cite{Luo2018Research}, insufficient reference nodes due to limited communication coverage of nodes \cite{Luo2018Research}, and the signal propagation path bending, which is called stratification effect, caused by the spatio-temporal variety of sound velocity \cite{Luo2021Review,jensen2011computational}.

\indent Comprehensive surveys of challenges and techniques in underwater localization have been done in \cite{Tan2011LocationSurvey,Qu2016LocattionSurvey,Luo2018Research,Luo2021Review}.  In 1991, James proposed a buoy-based long baseline positioning system in \cite{Youngberg1991GPStoUnderwater}, which extended the concept of global positioning system (GPS) to underwater systems for the first time. In 1994, Thomas developed an intelligent buoy underwater positioning system based on GPS in \cite{Thomas1994WaterGPS}. However, due to the special nature of underwater acoustic channels, high--precision positioning faces challenges mentioned above.

\indent To improve the reliability and accuracy of the positioning system, early positioning work mainly focused on the clock asynchronization problem for accurate underwater positioning. Cheng et~al. proposed a time of arrival (TOA) based silent underwater positioning scheme (UPS) in \cite{Cheng2008UPS}, which transforms the clock asynchronization problem between target node and reference nodes into the clock asynchronization problem among reference nodes. Since the clock difference between the reference node is known, the clock asynchronization problem is solved. However, due to the existence of multiple solutions in some areas, there are high requirements for the deployment position of anchor nodes. To make the positioning algorithm more universal, Cheng et~al. proposed a large--scale static node localization method named LSLS in \cite{Cheng2009LSLS}, in which the located node will serve as a new reference node, and gradually complete the positioning process of the whole ocean area. Liu et~al., Carroll et~al., and Yan et~al. respectively proposed localization algorithms in \cite{Liu2010ARTL,Carroll2014Ondemand,Yan2018Asynchronous} that utilize asymmetric signal round--trip processes to reduce clock synchronization requirements.

\indent As the development of urban positioning and indoor wireless positioning technology, many advanced signal processing technologies have emerged to weaken multipath signals or enhance the main path signals, such as sparse channel estimation in \cite{Berger2010SparseChannel,Gwun2013SparseChannel}, turbo equalization in \cite{Michael2011Turbo,Jung2013Turbo}, decision feedback equalization in \cite{Mahmutoglu2016PSO}, and time reversal mirror in \cite{Liu2019TimeReversal}, etc. Combining these signal processing methods, the impact of underwater acoustic multipath effect on underwater ranging can be effectively reduced. For example, Tan and Li proposed a centralized algorithm to overcome the severe multipath property of the underwater environment due to scattering from the seabed and ocean surface in \cite{Tan2010Cooperative}.

\indent With the scale expansion of underwater sensor networks, large-scale positioning has received widespread attention. However, due to the expensive underwater equipment and sparse node deployment, underwater positioning faces the problem of insufficient reference nodes. To tackle this problem, Zhou et~al. proposed an efficient localization algorithm for large--scale underwater sensor network in \cite{ZHOU2010EffLoc}, in which a multi--hop 3D Euclidean distance estimation method was proposed, and new reference nodes were searched among two hop neighbors. Based on \cite{ZHOU2010EffLoc}, Zhang et~al. proposed a top--down positioning scheme in \cite{Zhang2014ATP} that optimized the selection of new reference nodes. \cite{ZHOU2010EffLoc,Zhang2014ATP} provides positioning solution methods when the number of reference nodes is insufficient, but there is still a problem of positioning ambiguity (multi solution). Huang et~al. proposed an angle--of--arrival (AOA) assisted positioning method in \cite{Huang2016AOA}, which has no positioning ambiguity problem, but increases functional requirements for node equipment and raises production costs. Teymorian et~al.proposed a 3D localization algorithm in \cite{Teymorian2009USP}. When the depth information of the target is known, a projection method is proposed to project the reference node vertically onto the depth plane where the target node is located, thereby transforming the three-dimensional positioning problem into a two-dimensional positioning problem. Thus, there will be at least three reference nodes to complete underwater target positioning, which reduces requirement for the number of reference nodes. 

\indent Aforementioned algorithms have proposed good solutions for the problem of high latency, strong multipath, and limited node coverage for underwater positioning. However, they overlooked an important issue that the underwater sound velocity distribution exhibits spatio-temporal variability, resulting in a significant Snell effect in the signal propagation mode \cite{Tan2011LocationSurvey,Qu2016LocattionSurvey,Luo2021Review}. The strong Snell effect leads to the signal propagation path bending, which makes it difficult to accurately measure the signal propagation range. To reduce the accuracy loss of positioning caused by sound velocity estimation error, some works have been done to study how to correct the positioning error caused by sound velocity, which is called stratification compensation in this paper.

\indent Diamant and Lampe proposed a localization algorithm in \cite{Diamant2013Loc} that regards the sound velocity as an invariant parameter to be solved. Liu et~al. proposed a joint synchronization and localization method in \cite{Liu2016JSL}, Zhang et~al. proposed a signal trajectory correction localization method based on Gaussian Newton solution in \cite{Zhang2017Stratification}, Zhang et~al. proposed a Cramer-Rao lower bound based optimality localization (CRLB--OL) method in \cite{Zhang2018USC}, and Zhang et~al. combined TDOA and AOA measurements for localization with a constrained weighted least-squares estimator in \cite{Zhang2021EfficientUA}. \cite{Liu2016JSL,Zhang2017Stratification,Zhang2018USC,Zhang2021EfficientUA} have addressed the problem of stratification compensation, but mainly focused on system nodes with known target depth information. For non cooperative target node localization or nodes with damaged depth sensor, existing stratification compensation methods are difficult to apply.

\indent In order to accurately locate non cooperative targets or cooperative targets with damaged depth sensor, we propose a stratification compensation 3D underwater localization method based on iterative ray tracing (IRTUL). The core idea is to optimize the matching degree between simulated sound field and measured sound field based on ray--tracing theory, and iteratively correct the target location. To accelerate the positioning algorithm, we derive the signal propagation time and horizontal propagation distance in ray theory as a function of the initial grazing angle, and prove that they are both monotonically related to the initial grazing angle, so that the dichotomy method can be used for rapid searching of sound lines. Meanwhile, we propose a simplified expression method for the SVP, which reduces the sampling points and reduces the computational burden of ray tracing while ensuring the accuracy of signal ray tracing. The contribution of this paper is summarized as follows:
\begin{itemize}
	\item To accurately localize underwater target without known of depth information, we propose an IRTUL method, which can accurate locate targets without prior depth information.
	\item To demonstrate the feasibility of fast stratification compensation, we first derive the signal path as a function of glancing angle, then prove the executable of fast ray tracing.
	\item To accelerate the ray tracing process, we propose an SVP simplification method, which reduces the computational cost of ray tracing.
\end{itemize}

\indent The rest of this paper is organized as follows. In Sec. 2, related works about underwater acoustic localization are reviewed. In Sec. 3, we propose the IRTUL method for underwater target positioning, derive the signal path as a function of glancing angle, prove the executable of fast ray tracing, and give the fast achievement flowchart for IDAUL. Experimental results are discussed in Sec. 4, and conclusions are given in Sec. 5.
%%%%%%%%%%%%%%%%%%%%%%%%%%%%%%%%%%%%%%%%%%
\section{Related Works}
Youngberg proposed a buoy-based long baseline positioning system in \cite{Youngberg1991GPStoUnderwater}, which first extended the concept of global positioning system (GPS) to underwater systems, in which several buoys are used to form a long-baseline underwater acoustic localization system and provide navigation service for underwater users. In 1995, Thomas introduced a GPS assisted intelligent buoy underwater positioning system. Afterwards, commercial underwater acoustic positioning systems have been produced by various companies, such as Kongsberg Simrad in Norway, Link Quest in the United States, Nautronix in Australia, and Sonardyne in the United Kingdom.

\indent In order to improve the reliability and real-time performance of underwater acoustic positioning systems, many terrestrial positioning algorithms have been migrated to underwater positioning systems, however, the natural characteristics of underwater environments poses challenges for improving the accuracy and real-time performance of positioning algorithms. Erol et~al. conducted research on existing underwater acoustic sensor network architectures and positioning algorithms in \cite{erol2011survey}, analyzing the differences between underwater acoustic wireless sensor networks and terrestrial wireless sensor networks, as well as the impact of high attenuation, high latency, and other special characteristics of the marine environment on underwater positioning. Tan et~al. and Luo et~al. respectively analyzed the difficulties and challenges in underwater target positioning in \cite{Tan2011LocationSurvey} and \cite{Luo2021Review}, which pointed out the main factors affecting positioning accuracy, such as the change of sound velocity, clock asynchrony, and movement of sensor nodes.

\indent According to the positioning mode, localization algorithms can be divided into non--ranging and ranging--based positioning algorithms. The non-ranging localization algorithm mainly estimates the approximate position of targets based on the coverage area of reference nodes. Typical works include the positioning algorithms based on multi--hop distance such as DV--hop proposed by Niculescu and Nath in \cite{Niculescu2001ADHOC,Niculescu2003DV} and the density--aware hop--count localization (DHL) proposed by Wong et~al. in \cite{Wong2005density}, positioning algorithms dividing target areas based on the coverage range of signal arrival intensity such as the range-free localization (RLS) proposed by He et~al. in \cite{He2003RLS} and area localization scheme (ALS) proposed by Chandrasekhar and Seah in \cite{Chandrasekhar2006ALS}; positioning algorithms based on region segmentation with the assistance of underwater autonomous vehicles (AUVs) such as localization with mobile Beacons (LoMob) proposed by Lee and Kim in \cite{Lee2012LoMob} and localization with randomly moving AUV proposed by Zandi et~al. in \cite{Zandi2015AUV}; the dual-hydrophone localization algorithm proposed by Zhu et~al. in \cite{Zhu2016dualhydrophone},etc. The non-ranging positioning algorithms have the advantage of low communication energy consumption and long life periods, but the accuracy performance is difficult to meet the growing demand of underwater PNT systems.

\indent For accurate positioning, ranging-based positioning algorithms leveraging sound field observation information such as time of arrival (TOA), time difference of arrival (TDOA), angle of arrival (AOA), and received signal strength indicator (RSSI) have been widely studied. Early positioning algorithms mainly addressed the clock asynchronous problem, Cheng et~al. proposed an underwater positioning scheme (UPS) for static nodes in \cite{Cheng2008UPS}, which transforms the problem of clock asynchronization between the target node and reference nodes into the problem of clock asynchronization among reference nodes (the clock difference of reference nodes is known). The clock synchronization is not required, but there is a problem of multiple solutions in some areas, thus requiring high deployment positions for anchor nodes. Cheng et~al. expanded the application range of UPS and proposed a localization scheme for large scale underwater network (LSLS) in \cite{Cheng2009LSLS}, in which located nodes serve as new reference nodes for positioning other nodes. Liu et~al. and Carroll et~al. respectively proposed the asymmetrical round-trip localization (ARTL) and the on-demand asynchronous localization (ODAL) to reduce clock synchronization requirements in \cite{Liu2010ARTL} and \cite{Carroll2014Ondemand}. These positioning algorithms provide excellent solutions for dealing with asynchronous clock of underwater nodes, but require a certain number of static anchor nodes, which is difficult to be satisfied in some underwater systems.

\indent The localization of underwater nodes typically requires at least 4 reference nodes, however, due to the limited communication coverage of reference nodes, there may not be enough reference nodes for some nodes to be located. With the development of underwater unmanned platforms, Erol proposed a localization method based on dive and rise reference nodes (DNR) in \cite{Erol2007DNR} that introduces vertical mobile anchors to assist network node localization. Isik proposed an AUV-assisted 3D underwater localization (3DUL) algorithm in \cite{Isik2009DUL} that introduces AUV as moving anchors for network node localization. However, due to the accumulation of inertial navigation errors, the positioning accuracy of mobile nodes is insufficient. Teymorian proposed an equivalent 3D underwater sensor positioning (USP) algorithm in \cite{Teymorian2009USP}. In USP, a projection method is used to vertically project the reference node into the known depth plane (through depth sensors) where the target node is located. USP transforms the 3D positioning problem into a 2D positioning problem, so that the demand for the number of reference nodes is reduced. Inspired by USP, Kurniawan proposed a projection based underwater localization (PUL) algorithm in \cite{Kurniawan2013PUL} that sets up a series of virtual depth layer to optimize anchor node selection for large scale network localization. \cite{Teymorian2009USP} and \cite{Kurniawan2013PUL} improve the flexibility of positioning algorithm, but have a high dependence on depth sensors. 

\indent To solve the problem of large--scale network localization, Zhou et~al. proposed a multi--hop 3D Euclidean distance estimation method for large--scale localization (LSL) in \cite{ZHOU2010EffLoc} that searches for reference nodes among two--hop neighbors and iteratively completes the positioning process. Uddin proposed a localization technique for underwater sensor networks (LOTUS) in \cite{Uddin2016LOTUS} that provides the potential location area of the target node when the number of reference nodes is insufficient. \cite{ZHOU2010EffLoc} and \cite{Uddin2016LOTUS} provide positioning solutions under situations of insufficient reference nodes, but there is a problem of positioning ambiguity. Alexandri et~al. proposed a time difference of arrival target motion analysis (TD--TMA) method for positioning underwater vehicles in \cite{Alexandri2022MotionAnalysis}, which is able to effectively address the positioning ambiguity caused by the lack of reference anchors. Weiss et~al. proposed a semi-blind method for underwater localization in \cite{Weiss2022SemiBlind}, in which the refraction signals of ocean surface and bottom are adopted to achieve target localization without line of sight signals. Sun et~al. defined a second--order time difference of arrival to improve the precision of TDOA algorithm in \cite{Sun2020BBB}, and proposed a localization algorithm based on decision tree in \cite{Sun2022TreeLoc}, which both gain good accuracy of underwater positioning. Li et~al. proposed a localization approach based on the track-before-detect (TBD) to directly determine the target's position in \cite{Li2023TrackDetect}, which solves the difficulty of direct sound selection.

\indent The research works aforementioned provide various solutions for underwater positioning in different scenarios, and combined with advanced signal processing algorithms, the impact of clock asynchronization and limited reference node's coverage on positioning accuracy can be effectively reduced. However, due to the spatio--temporal variability of underwater sound velocity, there will be obvious stratification effect. Therefore, there may be significant positioning errors when assuming that the signal propagates in a straight line. To tackle this issue, Diamant and Lampe proposed a localization algorithm in \cite{Diamant2013Loc} that regards the sound velocity as a parameter to be solved, but the unknown sound velocity is assumed to be invariant, limiting its application scenarios of deep ocean coverage. Ramezani et~al. adopted the Time-of-Flight (ToF) measurements in \cite{Ramezani2013SSP} to compensate the stratification. Based on \cite{Ramezani2013SSP}, Liu et~al. proposed a joint synchronization and localization method in \cite{Liu2016JSL}, Zhang et~al. proposed a signal trajectory correction localization method based on Gaussian Newton solution in \cite{Zhang2017Stratification}, Zhang et~al. proposed a Cramer-Rao lower bound based optimality localization (CRLB--OL) method in \cite{Zhang2018USC}, and Zhang et~al. combined TDOA and AOA measurements for localization with a constrained weighted least-squares estimator in \cite{Zhang2021EfficientUA}. \cite{Liu2016JSL,Zhang2017Stratification,Zhang2018USC,Zhang2021EfficientUA} have addressed the problem of stratification compensation, but mainly aimed at system nodes with known target depth information, and their applications are limited for situations where the target depth information is unknown or inaccurate. Gong et~al. adopted the deep neural networks to detect and locate a mobile underwater target in \cite{Gong2020MLLoc}, however, only average sound speed could be estimated. Inspired by \cite{Gong2020MLLoc}, Yan et~al. proposed a broad-learning-based localization algorithm with isogradient SVPs in \cite{Yan2023BroadLearning}, but faces the same problem with \cite{Liu2016JSL,Zhang2017Stratification,Zhang2018USC,Zhang2021EfficientUA} when the depth is unknown. The deep--learning--based algorithms in \cite{Gong2020MLLoc} and \cite{Yan2023BroadLearning} have historical memory to update position output, but suffer from a time-consuming training process due to a large number of parameters in filters and layers.

\indent Stratification effect seriously affects the accuracy of underwater localization and has received widespread attention. However, existing works with stratification compensation rely on the depth information measured by depth sensors. When the depth information is unknown or inaccurate, it is a challenge to accurately locate the target.

%%%%%%%%%%%%%%%%%%%%%%%%%%%%%%%%%%%%%%%%%%
\section{Iterative Ray-tracing-based Underwater Acoustic Localization}
\subsection{System Model}
To solve the multi-solution issue during ray-tracing, we propose an IRTUL method, in which the 3D underwater localization problem is separated into multiple 2D localization problems. The localization system construction is shown in Fig.\ref{fig01}

\begin{figure}[!htbp]
	\centering
	\includegraphics[width=0.8\linewidth]{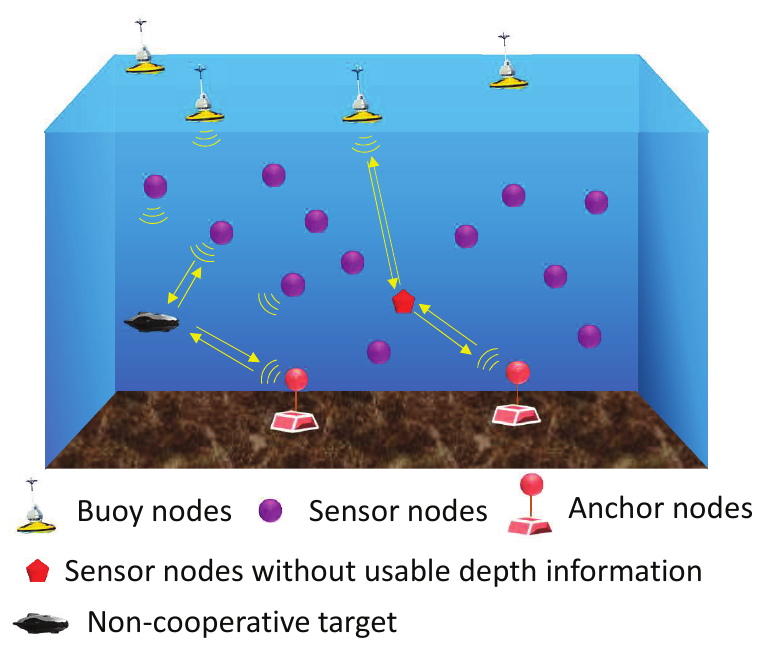}
	\caption{Localization system.}
	\label{fig01}
\end{figure}

\indent In this paper, we introduce our positioning method based on underwater observation networks. The localization system consists of 3 types of reference nodes: buoy nodes, sensor nodes, and anchor nodes. The position of buoy nodes is obtained by GPS and updated in real--time. Anchor nodes are fixed by carrying heavy objects and sinking to the seabed, which only need to be located when they are deployed. Sensor nodes are periodically located with the help of buoy and anchor nodes. There are 2 kinds of nodes to be located in our system, one is the sensor node whose depth information could not be directly measured through depth sensors, the other is the non--cooperative target whose depth information is an unknown parameter. The sensor node can communicate and interact with the reference node to achieve distance measurement, while non cooperative targets need to use detection signals for round-trip distance measurement.

\begin{figure}[!htbp]
	\centering
	\includegraphics[width=\linewidth]{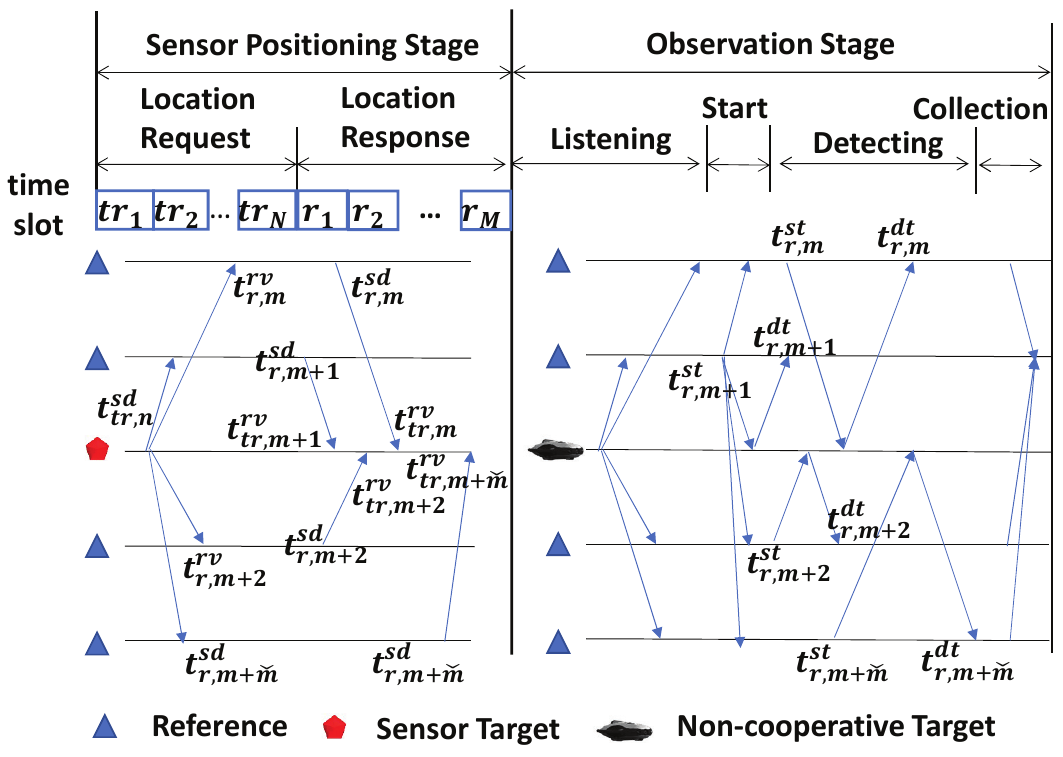}
	\caption{Positioning stage division.}
	\label{fig02}
\end{figure}

\indent There are two stage in the underwater observation networks as shown in Fig. \ref{fig02}, one is the sensor positioning stage, during which the sensor target is located, and it can be further divided as location request stage and location response stage. The other is the observation stage, during which the non-cooperative target is located, and it can be further divided into four stage: listening, starting, detecting, and collecting. Both the sensor target and the non-cooperative target are located through the round--trip TOA to reduce ranging errors caused by clock asynchronization.

\subsubsection{Sensor node localization}
\indent During the sensor positioning stage, time division multiplexing is adopted to complete the transmission of messages between nodes by considering the sparse distribution of underwater nodes. During the location request stage, sensor target node that needs to be located broadcast the location request information in the preset time slot according to the node number order. Then, in the location response stage, reference nodes also answer the location request in sequence by node number. Once a sensor node is located, it could act as a new reference node.

\indent Assume the position of the $n$th sensor target is $P_{tr,n}=\left(x_{tr,n},y_{tr,n},z_{tr,n}\right)$, and the position of the $m$th reference node is $P_{r,m}=\left(x_{r,m},y_{r,m},z_{r,m}\right)$. At time $t_{tr,n}^sd$, the sensor target node asks for positioning, and its messages arrive at reference node $m$ at time $t_{r,m}^{rv}$. The single--trip distance observation value will be:
\begin{equation}
	\rho_{m,n}^{tr,r} = f\left(P_{tr,n},P_{r,m}\right)+e_{t}+e_{s}+e_{n},\label{eq1}
\end{equation}
where $f\left(P_{tr,n},P_{r,m}\right)$ is the real distance between sensor target $n$ and reference node $m$, $e_{t}$ is the ranging error caused by clock asynchronization, $e_{s}$ is the ranging error caused by using constant sound speed value as 1500 $m/s$, and $e_{n}$ is the error caused by environmental noise.

\indent Without losing generality, let the reference node answer the location request at time $t_{r,m}^{sd}$, and the answer message arrive at the sensor target at time $t_{tr,m}^{sd}$, the $n$th sensor target's clock lag behind the reference node $m$ with an amount of $\Delta t_{tr,r}^{n,m}$, then the round--trip ranging will be:
\begin{equation}
\begin{split}
	2\rho_{m,n}^{tr,r} = c\left(t_{r,m}^{rv}-\left(t_{tr,n}^{sd}-\Delta t_{tr,r}^{n,m}\right)\right)+\\c\left(t_{tr,m}^{rv}-\Delta t_{tr,r}^{n,m}-t_{r,m}^{sd}\right)+2e_{s}+2e_{n},\label{eq2}
\end{split}
\end{equation}
Let $\delta_s = 2e_s$, $\delta_n = 2e_n$, equation \eqref{eq2} will be:
\begin{equation}
		2\rho_{m,n}^{tr,r} = 2f\left(P_{tr,n},P_{r,m}\right)+\delta_s+\delta_n,\label{eq3}
\end{equation}
where the clock asynchronization error could be eliminated through round--trip TOA.

\indent Because the depth information of the sensor target is unknown in this paper, there should be at least 4 reference nodes to finish the localization. Assume the number of reference nodes that lays within the communication coverage of the sensor target is $\breve{m}$, there will be:
\begin{strip}
	\begin{equation}
		\begin{cases}
			\sqrt{\left(\hat x_{tr,n} - x_{r,m}\right)^2+\left(\hat y_{tr,n} - y_{r,m}\right)^2+\left(\hat z_{tr,n} - z_{r,m}\right)^2}=\rho_{m,n}^{tr,r}\\
			\sqrt{\left(\hat x_{tr,n} - x_{r,m+1}\right)^2+\left(\hat y_{tr,n} - y_{r,m+1}\right)^2+\left(\hat z_{tr,n} - z_{r,m+1}\right)^2}=\rho_{m+1,n}^{tr,r}\\
			...\\
			\sqrt{\left(\hat x_{tr,n} - x_{r,m+\breve m}\right)^2+\left(\hat y_{tr,n} - y_{r,m+\breve m}\right)^2+\left(\hat z_{tr,n} - z_{r,m+\breve m}\right)^2}=\rho_{m+\breve m,n}^{tr,r}
		\end{cases}\label{eq4}
	\end{equation}
\end{strip}
where $\hat P_{tr,n}=\left(\hat x_{tr,n},\hat y_{tr,n},\hat z_{tr,n}\right)$ is the estimated location of sensor target node. The optimization objective of the positioning model will be:
\begin{equation}
	\hat P_{tr,n} = \underset {P_{tr,n}} {\text{arg min}} \sum_{1}^{\breve m} \left(2\rho_{m+\breve m,n}^{tr,r}-2f\left(\hat P_{tr,n},P_{r,m+\breve m}\right)\right).\label{eq5}
\end{equation}

\subsubsection{non-cooperative node localization}
\indent During the observation stage, the reference sensor node detect the non-cooperative target through active sonar to eliminate the time synchronization errors. All reference nodes are started being at listening stage. Once a random reference node hears the radiated noise, it will initiate a detection process as a temporary head node. The head node first send out a detection signal and use it to awaken neighboring reference nodes, which will respectively send their detection signals for target ranging.

\indent Different from sensor node localization, the detect signals will be immediately returned upon reaching the target surface, so there is no signal processing delay on the target node, and the time asynchronization error could also be eliminated through the round--trip transmission similar to equation \eqref{eq2}. Let the $(m+1)$th reference node be the head node, the detection signal is sent out at time $t_{r,m+1}^{st}$, and return at time $t_{r,m+1}^{dt}$, then the round--trip ranging will be:
\begin{equation}
	2\rho_{m+1,n}^{nctr,r} = 2c\left(t_{r,m+1}^{dt}-t_{r,m+1}^{st}\right)+2e_s+2e_n,\label{eq6}
\end{equation}
where $nctr$ means non-cooperative target. 

\indent Let the real position of the non-cooperative target is $P_{nctr,n}=\left(x_{nctr,n},y_{nctr,n},z_{nctr,n}\right)$, where $n$ represent the $n$th non-cooperative target, equation \eqref{eq6} can also be rewritten as equation \eqref{eq3}. While the optimization object of the positioning model will be:
\begin{equation}
	\hat P_{nctr,n} = \underset {P_{nctr,n}} {\text{arg min}} \sum_{1}^{\breve m} \left(2\rho_{m+\breve m,n}^{nctr,r}-2f\left(\hat P_{nctr,n},P_{r,m+\breve m}\right)\right).\label{eq7}
\end{equation}
where $\hat P_{nctr,n}=\left(\hat x_{nctr,n},\hat y_{nctr,n},\hat z_{nctr,n}\right)$ is the estimated location of non-cooperative target node.

\indent Obviously, directly solving equations \eqref{eq5} and \eqref{eq7} based on \eqref{eq4} could not reduce the ranging error caused by the variety of sound speed, and it needs further ranging correction through ray tracing theory.

\subsection{IRTUL for Target without Depth Information}
\indent For a given SVP $\mathcal{S}=[(s_0,d_0),(s_1,d_1),...,(s_i,d_i)],i=0,1,...I$, where $d_i$ is the (i+1)th sampling depth, $s_i$ is the corresponding sound speed value, the signal propagation time and horizontal distance are provided in \cite{Wang2013UnderwaterAcoustics} as:

\begin{equation}
	t = \sum_{i=1}^{I}\left|\frac{1}{g_{i}}ln\frac{tan(\frac{\theta_i}{2}+\frac{\pi}{4})}{tan(\frac{\theta_{i-1}}{2}+\frac{\pi}{4})}\right|,\label{eq8}
\end{equation}
\begin{equation}
	h = \frac{s_0}{cos\theta_0} \sum_{i=1}^{I}\left|\frac{sin\theta_{i-1}-sin\theta_i}{g_{i}}\right|,\label{eq9}
\end{equation}
where $s_0$ is the sound speed value of the first sampling depth, $\theta_0$ is the initial grazing angle, and $g_i$ is the gradient of sound speed that satisfies:
\begin{equation}
	s_i=s_{i-1} + g_i(d_i-d_{i-1}).\label{eq10}
\end{equation}
An example of signal propagation is shown in Fig.~\ref{fig03}

\setcounter{TempEqCnt}{\value{equation}} % 将当前公式序号 赋给TempEqCnt
\setcounter{equation}{18} % 当前公式序号变为x，x等于长公式应有的序号减1.
\begin{figure*}[ht]
	\begin{equation}
		\begin{cases}
			\left(\hat h_{m+1,n}^{ir_b}\right)^2 - \left(\hat h_{m,n}^{ir_b}\right)^2 = -2\hat x_{tr,n}^{ir_b} \left(x_{r,m+1}-x_{r,m}\right) -2\hat y_{tr,n}^{ir_b} \left(y_{r,m+1}-y_{r,m}\right) + x_{r,m+1}^2 - x_{r,m}^2 + y_{r,m+1}^2 - y_{r,m}^2\\
			\left(\hat h_{m+2,n}^{ir_b}\right)^2 - \left(\hat h_{m,n}^{ir_b}\right)^2 = -2\hat x_{tr,n}^{ir_b} \left(x_{r,m+2}-x_{r,m}\right) -2\hat y_{tr,n}^{ir_b} \left(y_{r,m+2}-y_{r,m}\right) + x_{r,m+2}^2 - x_{r,m}^2 + y_{r,m+2}^2 - y_{r,m}^2\\
			...\\
			\left(\hat h_{m+\breve m,n}^{ir_b}\right)^2 - \left(\hat h_{m,n}^{ir_b}\right)^2 = -2\hat x_{tr,n}^{ir_b} \left(x_{r,m+\breve m}-x_{r,m}\right) -2\hat y_{tr,n}^{ir_b} \left(y_{r,m+\breve m}-y_{r,m}\right) + x_{r,m+\breve m}^2 - x_{r,m}^2 + y_{r,m+\breve m}^2 - y_{r,m}^2
		\end{cases}.\label{eq19}
	\end{equation}
\end{figure*}
\setcounter{equation}{20}
\begin{figure*}[ht]
	\begin{equation}
		A=\begin{bmatrix}
			x_{r,m+1}^2 - x_{r,m}^2 + y_{r,m+1}^2 - y_{r,m}^2 + \left(\hat h_{m,n}^{ir_b}\right)^2 - \left(\hat h_{m+1,n}^{ir_b}\right)^2 \\
			x_{r,m+2}^2 - x_{r,m}^2 + y_{r,m+2}^2 - y_{r,m}^2 + \left(\hat h_{m,n}^{ir_b}\right)^2 - \left(\hat h_{m+2,n}^{ir_b}\right)^2\\
			...\\
			x_{r,m+\breve m}^2 - x_{r,m}^2 + y_{r,m+\breve m}^2 - y_{r,m}^2 + \left(\hat h_{m,n}^{ir_b}\right)^2 - \left(\hat h_{m+\breve m,n}^{ir_b}\right)^2
		\end{bmatrix}.\label{eq21}
	\end{equation}
\end{figure*}
\setcounter{equation}{10}

\begin{figure}[!htbp]
	\centering
	\includegraphics[width=0.8\linewidth]{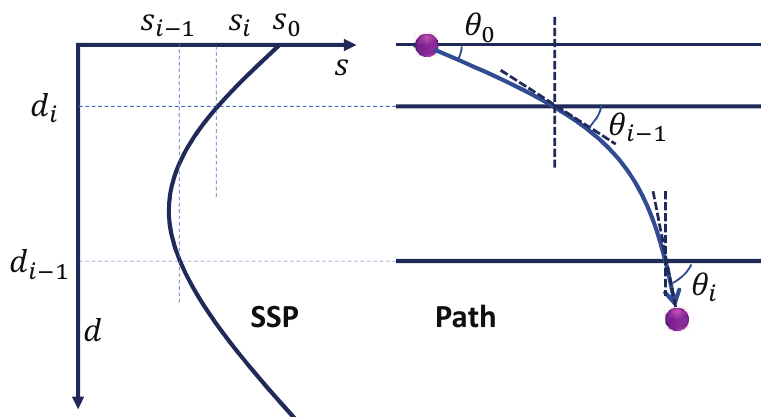}
	\caption{An example of signal propagation.}
	\label{fig03}
\end{figure}

According to the Snell's law of refraction:
\begin{equation}
n_{i-1}cos\theta_{i-1}=n_icos\theta_i,i=1,2,...,I,\label{eq11}
\end{equation}
\begin{equation}
\frac{s_i}{s_{i-1}}=\frac{n_{i-1}}{n_i},i=1,2,...,I,\label{eq12}
\end{equation}
where $n_i$ is the refractive index at the $i$th depth layer. There will be:
\begin{equation}
	\frac{s_i}{s_{i-1}}=\frac{cos_{i}}{cos_{i-1}},i=1,2,...,I.\label{eq13}
\end{equation}
Considering the trigonometric relationship:
\begin{equation}
	sin\theta_i = \sqrt{1-cos^2\theta_i},\label{eq14}
\end{equation}
\begin{equation}
	tan\left(\frac{\theta_i}{2}+\frac{\pi}{4}\right)=\frac{1+sin\theta_i}{cos\theta_i}.\label{eq15}
\end{equation}
The equation \eqref{eq8} and \eqref{eq9} can be respectively derived as:
\begin{equation}
	t = \sum_{i=1}^{I}\left|\frac{\Delta d_i}{s_{i}-s_{i-1}}ln\left(\frac{s_{i-1}}{s_i}\frac{1+\sqrt{\tau_i}}{1+\sqrt{\tau_{i-1}}}\right)\right|,\label{eq16}
\end{equation}
\begin{equation}
	h = \frac{s_0}{cos\theta_0} \sum_{i=1}^{I}\left|\frac{\Delta d_i}{s_{i}-s_{i-1}}\left(\sqrt{\tau_{i-1}}-\sqrt{\tau_{i}}\right)\right|,\label{eq17}
\end{equation}
where $\Delta d_i = d_i-d_{i-1}$,$\tau_{i} = 1-\left(\frac{s_i}{s_0}\right)^2cos^2\theta_0$. The equation \eqref{eq16} and \eqref{eq17} indicates that the signal propagation time and horizontal propagation distance can be expressed as a function of the initial grazing angle. However, the prerequisite is that the depth range of signal propagation is given. When the transmission depth is not prior information, the application of equation \eqref{eq16} and \eqref{eq17} will be restricted.

\begin{figure*}[!htbp]
	\centering
	\includegraphics[width=0.7\linewidth]{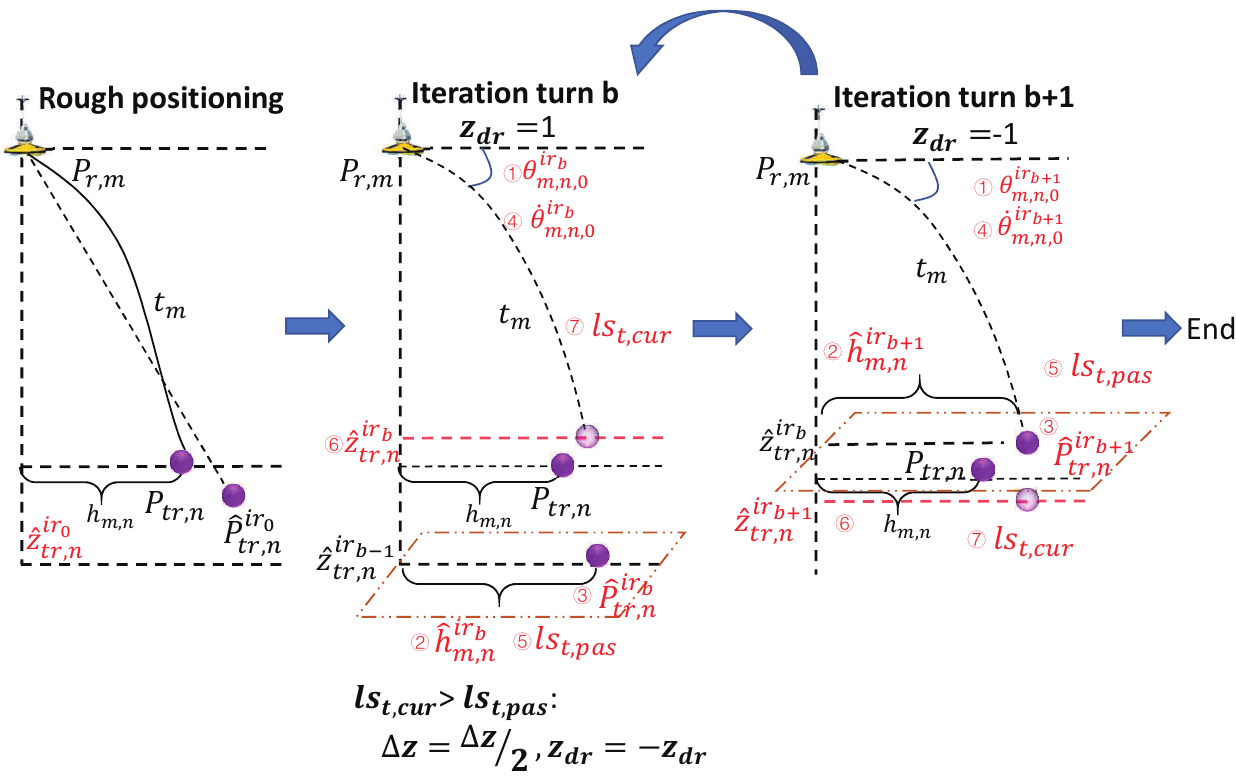}
	\caption{Scheme of IRTUL.}
	\label{fig04}
\end{figure*}
\indent To achieve accurate positioning, we propose the IRTUL, the scheme of which is shown in Fig.~\ref{fig04}. At the beginning, target node is roughly located as $\hat{P}_{tr,n}^{ir_0}$ with a linear propagation model based on the average ocean sound speed 1500 $m/s$. When coming into the $b$th turn of iteration, based on the current positioning depth $\hat{z}_{tr,n}^{ir_{b-1}}$ of the target node, the 1st step is to search the initial grazing angle $\theta_{m,n,0}^{ir_b}$ of the signal through equation \eqref{eq16} that exactly takes time $t_m$ to be transmitted from the reference node $P_{r,m}$ to the current target depth. The 2nd step is to calculate the horizontal transmitted distance $\hat h_{m,n}^{ir_b}$ of the signal through equation \eqref{eq17}. Assume there are $\breve m +1 $ reference nodes, then similar to equation \eqref{eq4}, there will be:

\begin{equation}
	\begin{cases}
		\left(\hat x_{tr,n}^{ir_b} - x_{r,m}\right)^2+\left(\hat y_{tr,n}^{ir_b} - y_{r,m}\right)^2= \left(\hat h_{m,n}^{ir_b}\right)^2\\
		\left(\hat x_{tr,n}^{ir_b} - x_{r,m+1}\right)^2+\left(\hat y_{tr,n}^{ir_b} - y_{r,m+1}\right)^2= \left(\hat h_{m+1,n}^{ir_b}\right)^2\\
		\quad \quad \quad \quad \quad \quad \quad \quad...\\
		\left(\hat x_{tr,n}^{ir_b} - x_{r,m+\breve m}\right)^2+\left(\hat y_{tr,n}^{ir_b} - y_{r,m+\breve m}\right)^2= \left(\hat h_{m+\breve m,n}^{ir_b}\right)^2
	\end{cases}.\label{eq18}
\end{equation}

\indent Subtract the first term from each sub equation in equation \eqref{eq18}, we can get equation \eqref{eq19} and \eqref{eq20}:

\setcounter{equation}{19}
\begin{equation}
	A=B\hat P_{tr,n,h}^{ir_b},\label{eq20}
\end{equation}
where $\hat P_{tr,n,h}^{ir_b} = \left(\hat x_{tr,n}^{ir_b},\hat y_{tr,n}^{ir_b}\right)$, $A$ and $B$ satisfies:

\setcounter{equation}{21}
\begin{equation}
	B=\begin{bmatrix}
		2\left(x_{r,m+1}-x_{r,m}\right) & 2\left(y_{r,m+1}-y_{r,m}\right)\\
		2\left(x_{r,m+2}-x_{r,m}\right) & 2\left(y_{r,m+2}-y_{r,m}\right)\\
		...\\
		2\left(x_{r,m+\breve m}-x_{r,m}\right) & 2\left(y_{r,m+\breve m}-y_{r,m}\right)
	\end{bmatrix}.\label{eq22}
\end{equation}

The equation \eqref{eq20} can be solved through the least square (LS) method:

\begin{equation}
	\hat P_{tr,n,h}^{ir_b} = \left(B^TB\right)^{-1}B^TA.\label{eq23}
\end{equation}

\indent The 3rd step is to form the position of target node $\hat P_{tr,n}^{ir_b} = \left(\hat x_{tr,n}^{ir_b},\hat y_{tr,n}^{ir_b},\hat z_{tr,n}^{ir_{b-1}}\right)$ after horizontal distance correction. At the 4th step, the new initial grazing angle $\dot{\theta}_{m,n,0}^{ir_b}$ is searched according to equation \eqref{eq17}. Then the signal propagation time $t_{pas,m}$ before depth tuning is simulated according to \eqref{eq16}, and time lose function is defined by the 5th step as:

\begin{equation}
	ls_{t,pas} = \sum_{m=1}^{1+\breve m}\left(t_{pas,m} - t_m\right)^2,\label{eq24}
\end{equation}
where $t_m$ is the measured signal propagation time related to reference node $m$. At 6th step, the depth will be adjusted by $\Delta z*z_{dr}$, where $z_{dr}=1$ at the beginning. The step size is controlled by $\Delta z$, while the direction is controlled by $z_{dr}$. After depth tuning, the final signal propagation time $t_{cur,m}$ is simulated again and the final lose function is calculated as:

\begin{equation}
	ls_{t,cur} = \sum_{m=1}^{1+\breve m}\left(t_{cur,m} - t_m\right)^2.\label{eq25}
\end{equation}

\indent Through comparison of $ls_{t,cur}$ and $ls_{t,pas}$, the depth tuning step and direction will be modified. For example, if $ls_{t,cur} > ls_{t,pas}$, it is indicated that the sound field matching error is increased after depth tuning, thus $\Delta z = \Delta z/2$ and $z_{dr} = -1$. The aforementioned process will continue until $ls_{t,cur} - ls_{t,pas} < \epsilon$, where $\epsilon$ is a preset threshold.

\subsection{Fast Achievement of Iterative Ray-tracing Localization}
The localization approach is based on the iterative idea, so the computational efficiency is quite important to achieve real--time positioning. In this section, we will first analysis the feasibility of fast ray tracing, then propose the fast IRTUL algorithm, and give the computational complexity.

\subsubsection{Feasibility Analysis of Fast Ray-tracing}
In IRTUL, there are two ways to speed up the algorithm, one is the setting the number of SVP's depth layers during once ray tracing, and the other is the signal ray searching method.

\paragraph{Simplification of the SVP}
\indent According to the derived ray tracing theory equation \eqref{eq16} and \eqref{eq17}, the complexity of calculating signal propagation time and horizontal propagation distance depends on the number of layers in the given SVP. Ray tracing is a model based on linear layering of SVP. Although the actual SVP is non--linear, when high--density sampling is performed on an SVP, the sound speed distribution between sampling points can be linearized to maintain high ray tracing accuracy. Therefore, if an SVP is reasonably linearized and simplified with a small number of sparse points, the computational cost during single sound ray tracing could be significantly reduced with a small amount of accuracy loss.

\begin{figure}[!htbp]
	\centering
	\includegraphics[width=0.8\linewidth]{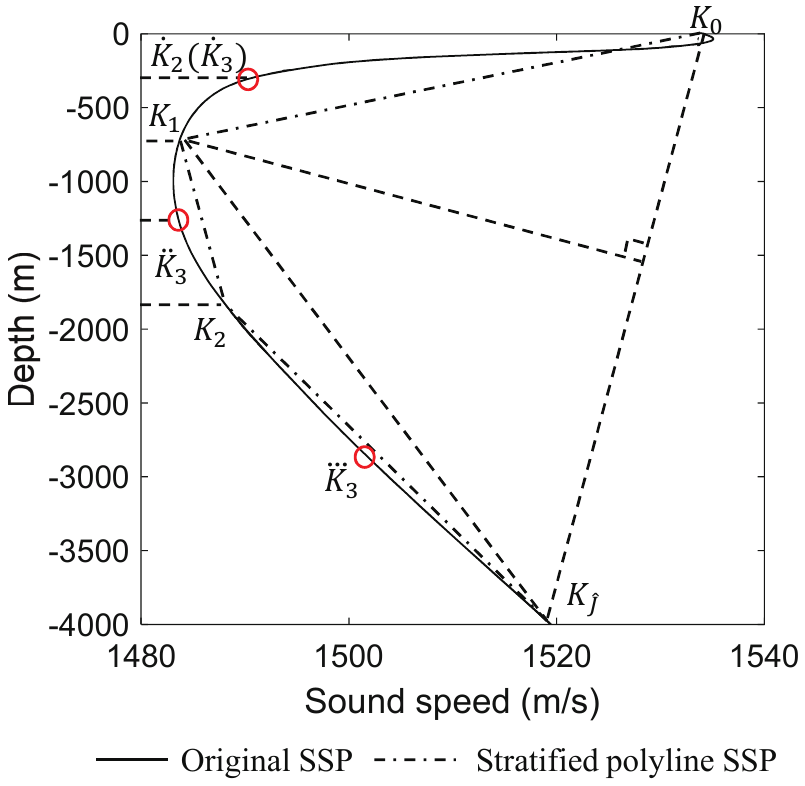}
	\caption{Simplification of SVPs.}
	\label{fig05}
\end{figure}

\indent To simplify a given SVP, we propose a simplification expression method for SVP based on maximizing distance reduction criterion as shown in Fig.~\ref{fig05}. During each turn of iteration, the candidate feature point will be searched in new sub--intervals, and the point with maximum distance to the current simplified SVP will be selected as new feature point, while the unselected point will be stored for next iteration. The convergence process of the algorithm is controlled by the number of preset feature points, or by comparing the root mean square error (RMSE) of the simplified SVP with the original SVP. The detailed introduction of the SVP simplification method can be found in our previous work by \cite{Huang2020SSPSimplification}, which is called distance-minimization-based (maximum
distance reduction) equal-interval control points searching (DM--EICPS). 

\paragraph{Optimization searching of signal grazing angle}
In our localization scheme, there are several steps in each iteration that require initial signal grazing angle searching. If it can be proven that the signal propagation time and horizontal distance are convex functions, then mature optimization searching methods can be used, so as to significantly improve computational efficiency compared to traversal searching method.

In this paper, we provide 3 properties to prove that the signal propagation time and horizontal propagation distance are monotonic functions of the initial grazing angle.

\newtheorem{property}{Property}
\begin{property}\label{1}
	For a linear layered SVP, the signal propagation path bends in the direction where sound speed decreases.
\end{property}

\begin{figure*}[!t]
	\centering
	\subfloat[]{\includegraphics[width=2.5in]{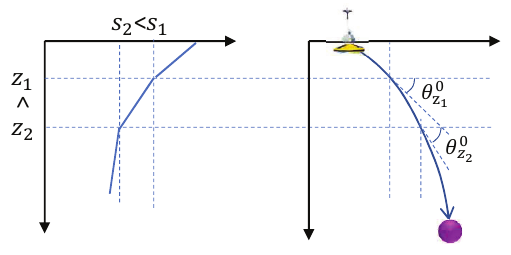}%
		\label{6(a)}}
	\hfil
	\subfloat[]{\includegraphics[width=2.5in]{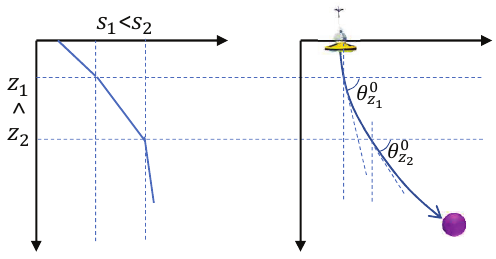}%
		\label{6(b)}}
	\caption{Bending direction of signal propagation path. (a) SVP with negative gradient. (b) SVP with positive gradient.}
	\label{fig6}
\end{figure*}

\begin{figure*}[!t]
	\centering
	\subfloat[]{\includegraphics[width=2.5in]{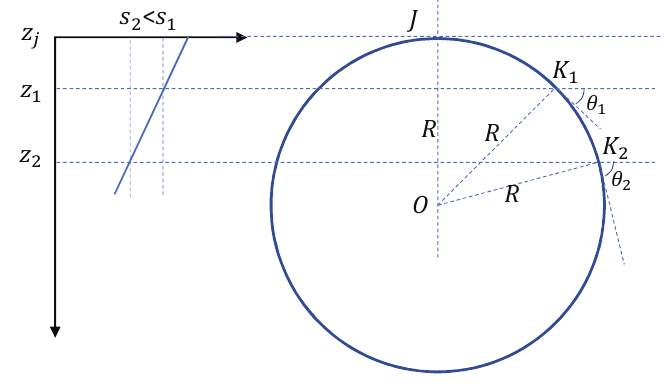}%
		\label{7(a)}}
	\hfil
	\subfloat[]{\includegraphics[width=2.5in]{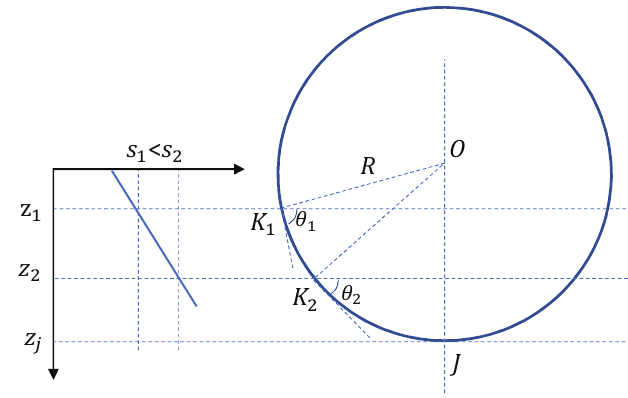}%
		\label{7(b)}}
	\caption{Proof of arc trajectory. (a) SVP with negative gradient. (b) SVP with positive gradient.}
	\label{fig7}
\end{figure*}

\begin{figure*}[!t]
	\centering
	\subfloat[]{\includegraphics[width=3.5in]{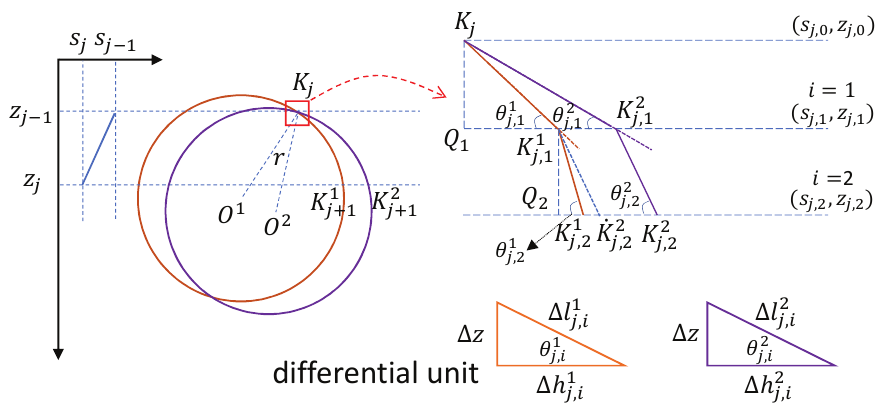}%
		\label{8(a)}}
	\hfil
	\subfloat[]{\includegraphics[width=2in]{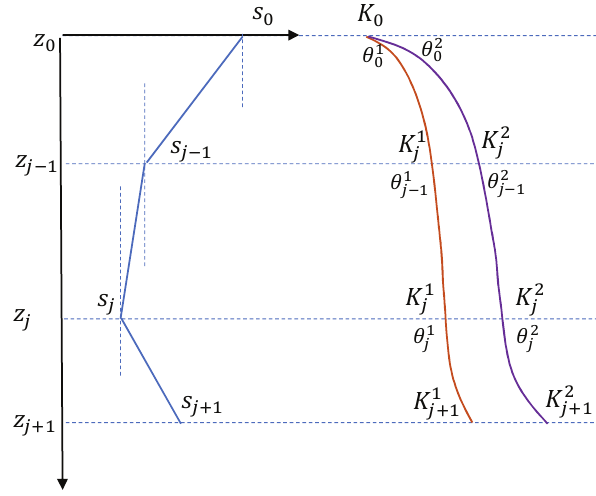}%
		\label{8(b)}}
	\caption{Proof of monotonicity. (a) Micro--element perspective. (b) Macro perspective.}
	\label{fig8}
\end{figure*}

\begin{proof} For linear layered SVP with negative gradient as shown in Fig.~\ref{6(a)} that $s_1>s_2$ and $z_1<z_2$, according to Snell's law of refraction, there is:
\begin{equation}
	s_1cos\theta_{z_2}^0 =s_2cos\theta_{z_1}^0.\label{eq26}
\end{equation}
where $cos\theta_{z_1}^0$ is the grazing angle at depth $z_1$. Since $s_1>s_2$, so $cos\theta_{z_2}^0 < cos\theta_{z_1}^0$. For $\theta_{z_2}^0,\theta_{z_1}^0 \in [0,90^\circ]$, there will be $\theta_{z_2}^0 > \theta_{z_1}^0$, which means the signal propagation path bends towards the direction where sound speed decrease (down). Similarly, for linear layered SVP with negative gradient as shown in Fig.~\ref{6(b)} that $s_1<s_2$ and $z_1<z_2$. According to equation \eqref{eq26}, since $s_1<s_2$, so $cos\theta_{z_1}^0 < cos\theta_{z_2}^0$. For $\theta_{z_2}^0,\theta_{z_1}^0 \in [0,90^\circ]$, there will be $\theta_{z_2}^0 < \theta_{z_1}^0$, which means the signal propagation path bends towards the direction where sound speed decrease (up).

\end{proof}

\begin{property}\label{2}
	For a linear layered SVP, the signal propagation path within any linear depth layer is an arc.
\end{property}

\begin{proof} (by contradiction)	As shown in Fig.~\ref{7(a)} that $s(z) = s_1 + g(z-z_1)$, where $g<0$ is the gradient of sound speed. According to property 1, the signal propagation path will be down--curved. The signal is sent out from $K_1$, and an auxiliary line in a direction perpendicular to the signal propagation direction at $K_1$ is set as $K_1O=R=-s_1/\left(gcos\theta_1\right)$. Let $K_2$ be a random point of the arc, and the angle between the tangent of the arc passing through points $K_1$ or $K_2$ and the horizontal direction are $\theta_1$ or $\theta_2$. 

\indent Assume that $K_2$ is not a point on the signal propagation trajectory, so there will be:
\begin{equation}
	\frac{s_1}{cos\theta_1}\neq\frac{s_2}{cos\theta_2}.\label{eq27}
\end{equation}
In fact, based on geometric relationships, there is:
\begin{equation}
	\begin{cases}
	z_j = z_1 + R cos\theta_1\\
	z_j = z_2 + R cos\theta_2	
	\end{cases}.\label{eq28}
\end{equation}
Substitute $R=-s_1/\left(gcos\theta_1\right)$, there will be:
\begin{equation}
	\frac{cos\theta_2}{cos\theta_1}=\frac{(z_2-z_1)g+s_1}{s_1}.\label{eq29}
\end{equation}
Since $s_2 = s_1 + g(z_2-z_1)$, there will be:
\begin{equation}
	\frac{cos\theta_2}{cos\theta_1}=\frac{s_2}{s_1},\label{eq30}
\end{equation}
which means equation \eqref{eq27} is not valid. Therefore, the signal propagation path $\widehat{K_1K_2}$ is a segment of circle arc. Similar conclusions could be drawn in Fig.~\ref{7(b)}.
\end{proof}

\newtheorem{corollary}{Corollary}
\begin{corollary}\label{3}
	For a linear layered SVP, the propagation time and horizontal propagation distance of non--reflected signals are monotonically decreasing functions of the initial grazing angle.
\end{corollary}

\begin{proof}{corollary 1:} Assuming that the gradient of an SVP is negative, and the signal shoots from shallow water to deep water. Divide the signal trajectory with different initial grazing angles into $I$ small segments for the $i$th depth layer as shown in Fig.~\ref{8(a)}. There is:
\begin{equation}
	z_{j-1}=z_{j,0}<z_{j,1}<...<z{j,i}<...<z_{j,I}=z_j, \label{eq31}
\end{equation}
where $i=0,1,...,I$. Since the gradient of sound speed is negative, so there is $s_{j,i} < s_{j,i-1}$. Let the depth is equally divided, and the space is $\Delta z$. If $\Delta z$ is small enough, the propagation path of the signal in this depth space can be approximated by a straight line, so the propagation time of the signal is:
\begin{equation}
	\Delta t_{j,i} = \frac{\Delta l_{j,i}}{\bar{s_i}}=\frac{\sqrt{\Delta \left(h_{j,i}\right)^2+\Delta z^2}}{\bar{s_i}}, \label{eq32}
\end{equation}
where $\bar{s_i}=0.5\left(s_{i-1}+s_{i}\right)$. According to the equivalence relationship between definite integral and summation limit:
\begin{equation}
	\begin{split}
	\begin{aligned}
	\Delta t_{j} &=\lim_{\Delta z\rightarrow0}\sum_{i=1}^{I}\Delta t_{j,i} \\
	&=\int_{K_j}^{K_{j+1}}dt_j=\int_{K_j}^{K_{j+1}}\frac{dl_j}{\bar{s_i}}\\
	&=\lim_{\Delta z\rightarrow0}\sum_{i=1}^{I}\frac{\Delta l_j}{\bar{s_i}} = \lim_{\Delta z\rightarrow0}\sum_{i=1}^{I}\frac{\sqrt{\Delta \left(h_{j,i}\right)^2+\Delta z^2}}{\bar{s_i}}
	\end{aligned}
	\end{split}.\label{eq33}
\end{equation}
For the two paths shown in Fig.~\ref{8(a)}, if $\Delta h_{j,i}^1<\Delta h_{j,i}^2$ is proven, then there will be $t_j^1 < t_j^2$ and $h_j^1 < h_j^2$.

\indent When $i=1$, let the grazing angles of two paths are $\theta_{j,1}^1$ and $\theta_{j,1}^2$, respectively, which satisfies $\theta_{j,1}^1>\theta_{j,1}^2$. Because $0<\theta_{j,1}^1,\theta_{j,1}^2<90^\circ$, the $\theta_{j,1}^1,\theta_{j,1}^2$ will satisfy $tan\theta_{j,1}^1>tan\theta_{j,1}^2$, then:

\begin{equation}
	\Delta h_{j,1}^1 = \frac{\Delta z}{tan\theta_{j,1}^1}<\frac{\Delta z}{tan\theta_{j,1}^2}=\Delta h_{j,1}^2, \label{eq34}
\end{equation}

\indent For $i=2$, according to the Snell's law:
\begin{equation}
	\frac{s_{j,1}}{s_{j,2}}=\frac{cos\theta_{j,1}^1}{cos\theta_{j,2}^1}=\frac{cos\theta_{j,1}^2}{cos\theta_{j,2}^2}. \label{eq35}
\end{equation}
Since $\theta_{j,1}^1>\theta_{j,1}^2$, so $cos\theta_{j,1}^1<cos\theta_{j,1}^2$ and $cos\theta_{j,2}^1<cos\theta_{j,2}^2$, which indicates that at the $2$nd depth layer, there is still $\theta_{j,2}^2>\theta_{j,2}^1$. When repeating the proof process for $i=2,3,4,...,I$, there will be $\Delta h_{j,i}^1<\Delta h_{j,i}^2$. Finally, we get $t_j^1 < t_j^2$ and $h_j^1 < h_j^2$.

\indent Similarly, it can be proven that for SVP layers with positive gradient, the propagation time and horizontal propagation distance of signals are also monotonically decreasing functions of the initial grazing angle.

\indent For SVPs with multiple linear layers as shown in Fig.~\ref{8(b)}, if $\theta_0^1 > \theta_0^2$, there will be $t_j^1 < t_j^2$ and $h_j^1 < h_j^2$ based on the above proof process. Finally, there is:
\begin{equation}
	t^1 = \sum_{j=1}^Jt_j^1< \sum_{j=1}^Jt_j^2 = t^2, \label{eq36}
\end{equation}
\begin{equation}
	h^1 = \sum_{j=1}^Jh_j^1< \sum_{j=1}^Jh_j^2 = h^2, \label{eq37}
\end{equation}
which proves the corollary 1.
\end{proof}

\subsubsection{Fast Algorithm of IRTUL}
Based on the aforementioned analysis, the initial grazing angle can be searched through dichotomy, and the ray tracking process can be accelerated through SVP simplification, which helps achieve the fast operation of IRTUL. In this section, we give the fast algorithm of IRTUL as shown in Algorithm 1.

\begin{algorithm}[H]
	\caption{Fast IRTUL Algorithm.}\label{alg:alg1}
	\begin{algorithmic}
		\STATE {\textsc{INPUT:}}
		\STATE \hspace{0.5cm}Simplified SVP $\tilde{\mathcal{S}}=[(\tilde{s_0},\tilde{d_0}),...,(\tilde{s_i},\tilde{d_i})],i=0,1,...,I$; reference nodes $P_{r,m} = (x_{r,m},y_{r,m},z_{r,m})$; measured signal propagation time $t_{m,n}$; threshold value of horizontal propagation distance $Th_h$; threshold value of signal propagation time $Th_t$; threshold of depth tuning step $Th_{\Delta z}$.
		\STATE {\textsc{Step 1:}}
		\STATE \hspace{0.5cm} Obtain ranging based on linear propagation model with average sound speed $\bar{s}$: $\rho_{m,n} = \bar{s}t_m$; conduct ball intersection positioning to obtain rough estimated position $\hat{P}_{tr,n}^{ir_0}$ with depth $z_{tr,n}^{ir_0}$; initialize the iteration as $b=1$.
		\STATE {\textsc{Step 2:}}
		\STATE \hspace{0.5cm} {\textsc{\textbf{WHILE}}} iteration b, and $\Delta z_b < TH_{\Delta z}$
		\STATE {\textsc{Step 3:}}
		\STATE \hspace{0.5cm} According to Equation \eqref{eq16}, use dichotomy to calculate the signal propagation time $t_{m,n}^{ir_0}$ for each reference node when signal is transmitted to depth $z_{tr,n}^{ir_b-1}$, when $\left\|t_{m,n}^{ir_0} - t_{m,n}\right\|_2<Th_t$ , record the corresponding signal grazing angle $\theta_{m,n,0}^{ir_b}$.
		\STATE {\textsc{Step 4:}}
		\STATE \hspace{0.5cm} Based on $\theta_{m,n,0}^{ir_b}$, calculate the horizontal propagation distance $\hat h_{m,n}^{ir_b}$ of signal according to \eqref{eq17}.
		\STATE {\textsc{Step 5:}}
		\STATE \hspace{0.5cm} Based on $\hat h_{m,n}^{ir_b}$, relocate the target at depth $z_{tr,n}^{ir_b-1}$ to get new position $\hat{P}_{tr,n}^{ir_b}$ with new horizontal distance as $\hat {\dot{h}}_{m,n}^{ir_b}$.
		\STATE {\textsc{Step 6:}}
		\STATE \hspace{0.5cm} Based on $\hat {\dot{h}}_{m,n}^{ir_b}$, search for the grazing angle $\dot \theta_{m,n,0}^{ir_b}$ by dichotomy when calculated distance $\hat {\ddot{h}}_{m,n}^{ir_b}$ according to equation \eqref{eq17} satisfies $\left\|\hat {\dot{h}}_{m,n}^{ir_b} -\hat {\ddot{h}}_{m,n}^{ir_b} \right\|_2<Th_h$.
		\STATE {\textsc{Step 7:}}
		\STATE \hspace{0.5cm} Based on $\dot \theta_{m,n,0}^{ir_b}$, calculate the signal propagation time $\hat{\dot t_{m,n}^{ir_b}}$ according to \eqref{eq16}.
		\STATE {\textsc{Step 8:}}
		\STATE \hspace{0.5cm} Calculate the lost function \eqref{eq24}.
		\STATE {\textsc{Step 9:}}
		\STATE \hspace{0.5cm} \textbf{If} $z_{dr} = 1$ \textbf{then} $\hat z_{tr,n}^{ir_b} = z_{tr,n}^{ir_{b-1}} + \Delta z$.
		\STATE \hspace{0.5cm} \textbf{Else} $\hat z_{tr,n}^{ir_b} = z_{tr,n}^{ir_{b-1}} - \Delta z$.
		\STATE {\textsc{Step 10:}}
		\STATE \hspace{0.5cm} Calculate the lost function \eqref{eq25}.
		\STATE {\textsc{Step 11:}}
		\STATE \hspace{0.5cm} \textbf{If} $ls_{t,cur}>ls_{t,pas}$ \textbf{then} $\Delta z = \Delta z/2, z_{dr} = -z_{dr}$.
		\STATE {\textsc{Step 12:}}
		\STATE \hspace{0.5cm} Record the current target position as $\hat{P}_{tr,n} = (\hat x_{tr,n},\hat y_{tr,n},\hat z_{tr,n})$ according to the smaller one between $ls_{t,cur}$ and $ls_{t,pas}$, $b=b+1$.
		\STATE \hspace{0.5cm} \textbf{END WHILE}
		\STATE {\textsc{OUTPUT:}} Final target position $\hat{P}_{tr,n} = (\hat x_{tr,n},\hat y_{tr,n},\hat z_{tr,n})$.
	\end{algorithmic}
	\label{alg1}
\end{algorithm}

%%%%%%%%%%%%%%%%%%%%%%%%%%%%%%%%%%%%%%%%%%
\section{Experimental Results}
\subsection{Parameter Settings}
To verify the accuracy and efficiency performance of the proposed IRTUL algorithm, we conduct simulation experiment in this section for testing. The simulation is conducted by Matlab (2023a). The simulation experimental parameters are shown in Table \ref{table1}. The system error for measuring signal propagation time is set to 0.1\% of the distance range, following zero-mean$^2$ white Gaussian nose with standard deviation $\sigma$ according to \cite{Liu2016JSL,Carroll2014Ondemand}. Thus when the average sound speed is assumed to be 1500 $m/s$, the equivalent signal propagation time error will be $\sigma_t = 0.003 s$. 
\begin{table}[!htbp]
	\caption{Parameter setting\label{table1}}
	\centering
	\begin{tabular}{|c||c|}
		\hline
		Item & Value\\
		\hline
		communication range & 4500 $m$\\
		\hline
		area square & 10 $km\times$ 10 $km$ \\
		\hline
		depth & 3 $km$\\
		\hline
		surface buoys & 25\\
		\hline
		anchor nodes & 25\\
		\hline
		target nodes to be located & 200\\
		\hline
		mean error (time) & 0\\
		\hline
		standard deviation error (time) $\sigma_t$ & 0.003 $s$\\
		\hline
		number of simplification layers & 7\\
		\hline
		threshold of depth tuning step & 0.2 $m$\\
		\hline
		threshold of signal propagation time $Th_t$ & 10 $\mu s$\\
		\hline
		threshold of horizontal propagation distance $Th_h$ & 0.1 $m$\\
		\hline
		initial depth step & 2 $m$\\
		\hline
	\end{tabular}
\end{table}

\begin{figure}[!htbp]
	\centering
	\includegraphics[width=0.8\linewidth]{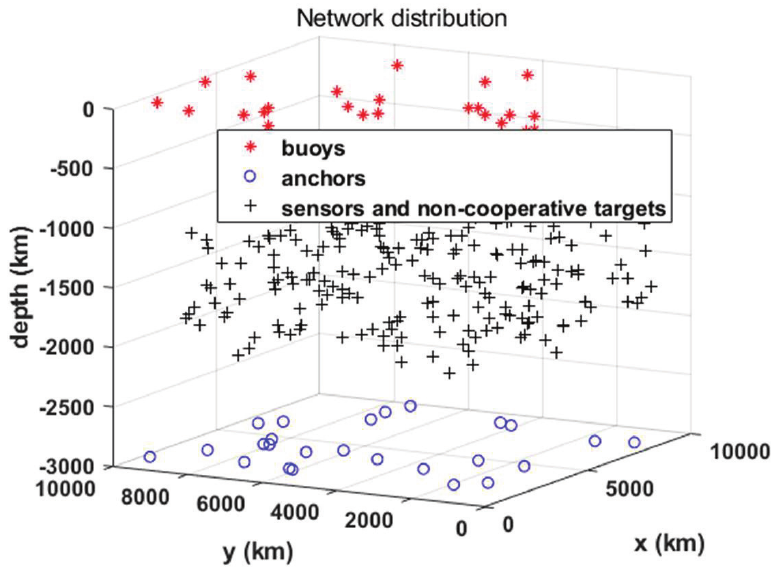}
	\caption{Simulation scenario of network.}
	\label{fig09}
\end{figure}

\indent The simulation scenario is provided in Fig.\ref{fig09}, where 25 surface buoys are uniformly distributed, and their positioning coordinates are obtained by GPS. 25 anchors are also uniformly distributed at the bottom, and their location is known (obtained through positioning by buoys). 200 nodes to be located are randomly suspended in water.

\subsection{Simulation Results}

\begin{figure}[!t]
	\centering
	\subfloat[]{\includegraphics[width=2.5in]{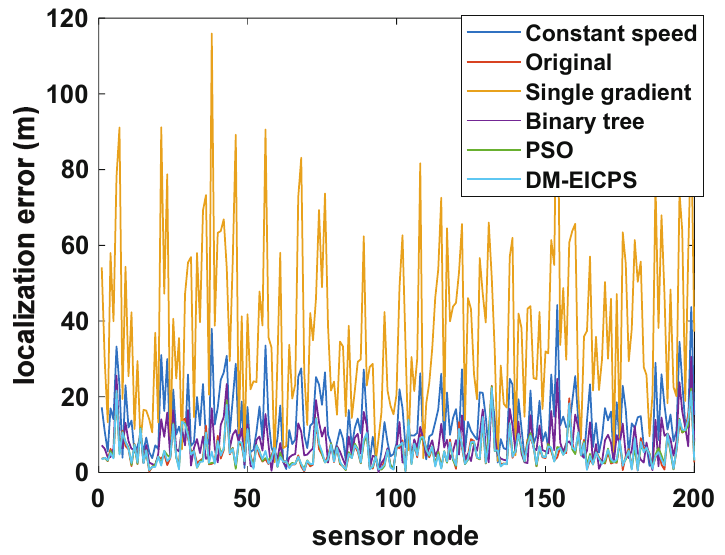}%
		\label{10(a)}}
	\hfil
	\subfloat[]{\includegraphics[width=2.5in]{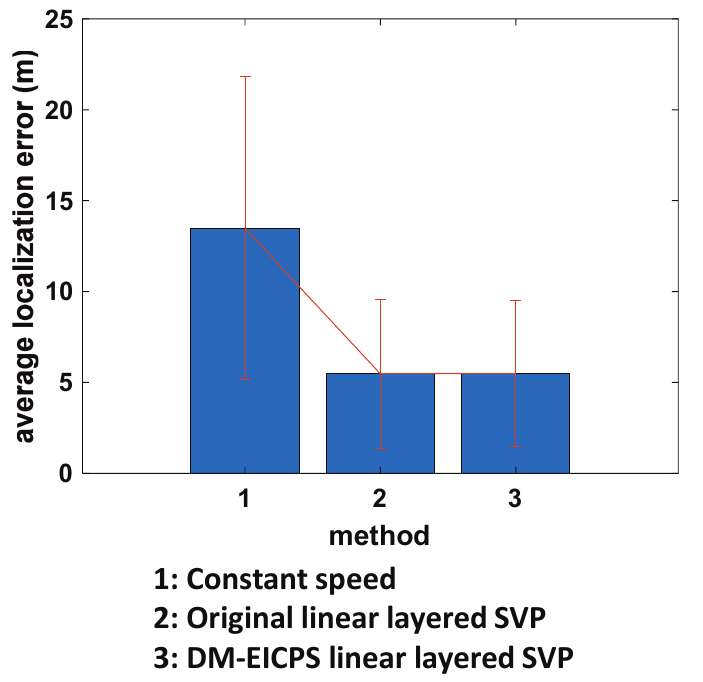}%
		\label{10(b)}}
	\caption{Accuracy perpformance with different kinds of SVP simplification method. (a) 200 nodes. (b) Average results.}
	\label{fig10}
\end{figure}

\indent To verify the feasibility of using simplified SVPs in positioning system, we provide localization accuracy results for different simplified methods of SVP as shown in Fig.~\ref{fig10}. The positioning RMSE errors of total 200 sensors or non-cooperative targets are given in Fig.~\ref{10(a)}, where the curve of proposed DM--EICPS is almost at the bottom, indicating that the simplified SVP by DM--EICPS for positioning has better accuracy performance. The average and strandard deviation of localization errors by constant speed, original linear SVP, and DM-EICPS simplified SVP is shown in Fig.~\ref{10(b)}. It shows that the positioning error of the simplified SVP is slightly greater than the result using the original SVP, but the difference is very small. If the number of layers in the simplified SVP is further increased, the positioning results will become closer to the positioning results by using the original SVP.

\begin{table*}[!htbp]
	\caption{Average positioning RMSE and increase in distance correction\label{table2}}
	\centering
	\begin{tabular}{|c||c||c||c||c||c||c||c||c|}
		\hline
		constant speed (m)& original SVP (m)& x (m)& y (m)& z (m)& SVP by DM--EICPS &  x (m)& y (m)& z (m)\\
		\hline
		8.435 & 5.591 & 1.499 & 2.189 & 5.087 & 5.676 & 1.959 & 2.789 & 5.134 \\
		\hline
	\end{tabular}
\end{table*}

\indent Table \ref{table2} gives the average localizaiton RMSE of 20 nodes and increase of distance correction in the direction of x, y, and z, respectively. The results show that compared to the positoning model with constant sound speed, IRTUL has the most significant distance correction in the depth direction. The average accuracy of IRTUL has been improved by about 3 $m$.

\begin{figure}[!t]
	\centering
	\subfloat[]{\includegraphics[width=2.5in]{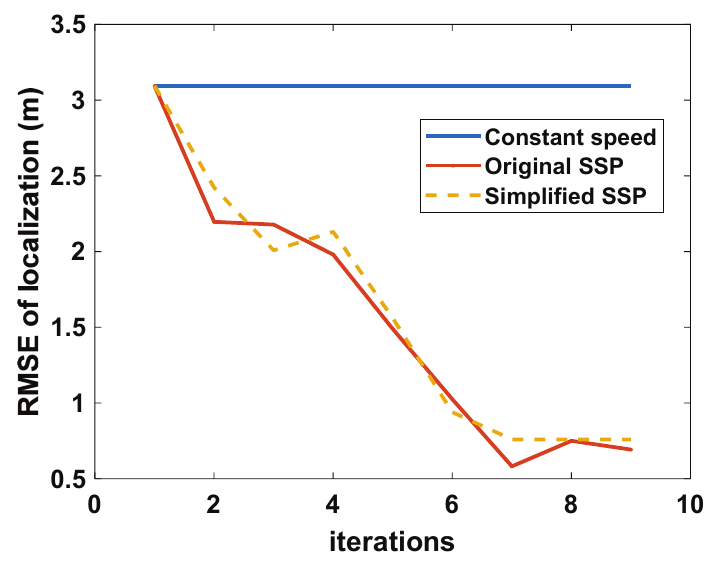}%
		\label{11(a)}}
	\hfil
	\subfloat[]{\includegraphics[width=2.5in]{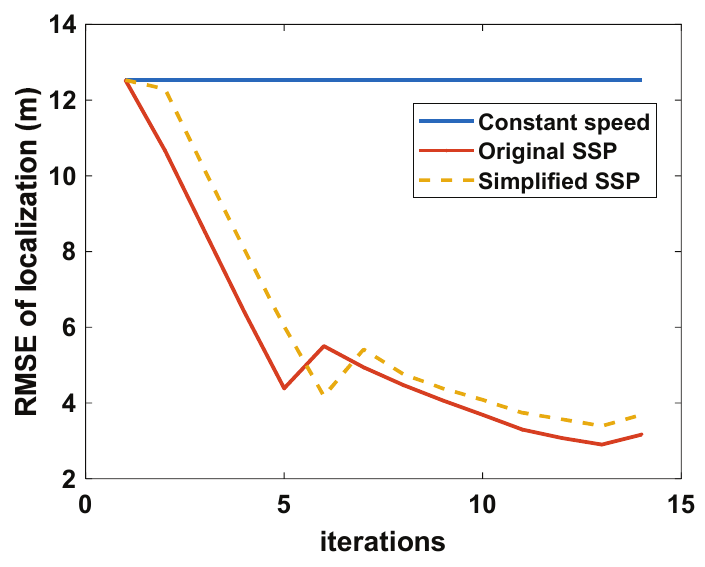}%
		\label{11(b)}}
	\caption{Convergence performance of ITRUL. (a) Node 1. (b) Node 2.}
	\label{fig11}
\end{figure}

\indent In order to verify the convergence stability of the IRTUL method, Fig.~\ref{fig11} shows the convergence of two sets of data. It can be seen that when the number of iterations exceeds 7 and 10, respectively in Fig.~\ref{11(a)} and Fig.~\ref{11(b)}, the positioning process has shown convergence, and the positioning accuracy has been significantly improved compared to the constant speed situation.

\begin{figure}[!htbp]
	\centering
	\includegraphics[width=0.8\linewidth]{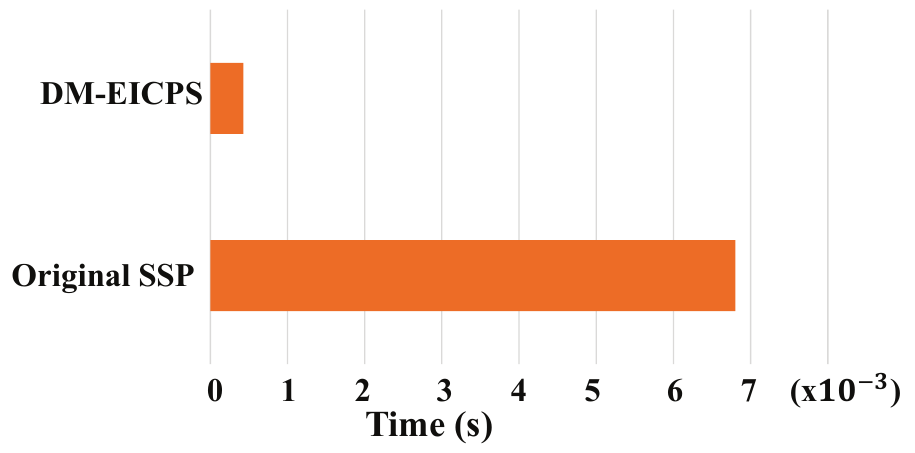}
	\caption{Time efficiency of ITRUL.}
	\label{fig12}
\end{figure}

Although the positioning accuracy using a simplified SVP is slightly lower than that using the original SVP, the former can greatly reduce the computational complexity of ray tracing and the time cost of the positioning correction algorithm. To verify the time effectiveness of the positioning correction algorithm using a simplified SVP, Fig.~\ref{fig12} compares the average time of positioning process using a simplified SVP with that using the original SVP under 10 test times. The results show that the time improvement effect is very significant.

%%%%%%%%%%%%%%%%%%%%%%%%%%%%%%%%%%%%%%%%%%
\section{Conclusion}
In this paper, we propose an IRTUL method for fast underwater target localization, which address the issue of sound line bending and positioning accuracy degradation caused by uniformly distributed sound speed. Firstly, the monotonic relationship between signal propagation time and horizontal propagation distance with respect to the initial grazing angle is proved, so that binary search can be used to quickly obtain the initial grazing angle of the sound line. Then, we establish the IRTUL model. By using ray tracing and depth iterative tuning, the real signal propagation path of the signal can be effectively tracked in the case of unknown target depth, which improves the positioning accuracy. To reduce the computational complexity of ray--theory--based sound correction process, an adaptive linear simplification method for SVPs based on minimizing the maximum distance criterion is proposed as DM--EICPS. Compared with traditional simplification method for contour feature extraction, DM-EICPS exhibits higher positioning accuracy while significantly reduce the computational cost compared to the positioning with original SVPs.
 
%%%%%%%%%%%%%%%%%%%%%%%%%%%%%%%%%%%%%%%%%%
\section*{Acknowledgments}
This work was supported by Natural Science Foundation of Shandong Province (ZR2023QF128), China Postdoctoral Science Foundation (2022M722990 and 2022M723888), Qingdao Postdoctoral Science Foundation (QDBSH20220202061), National Natural Science Foundation of China (62271459), National Defense Science and Technology Innovation Special Zone Project: Marine Science and Technology Collaborative Innovation Center (22-05-CXZX-04-01-02), Central University Basic Research Fund of China, Ocean University of China (202313036).

\bibliographystyle{IEEEtran}
\bibliography{IEEEabrv,draft_hw}

\begin{IEEEbiography}[{\includegraphics[width=1in,height=1.25in,clip,keepaspectratio]{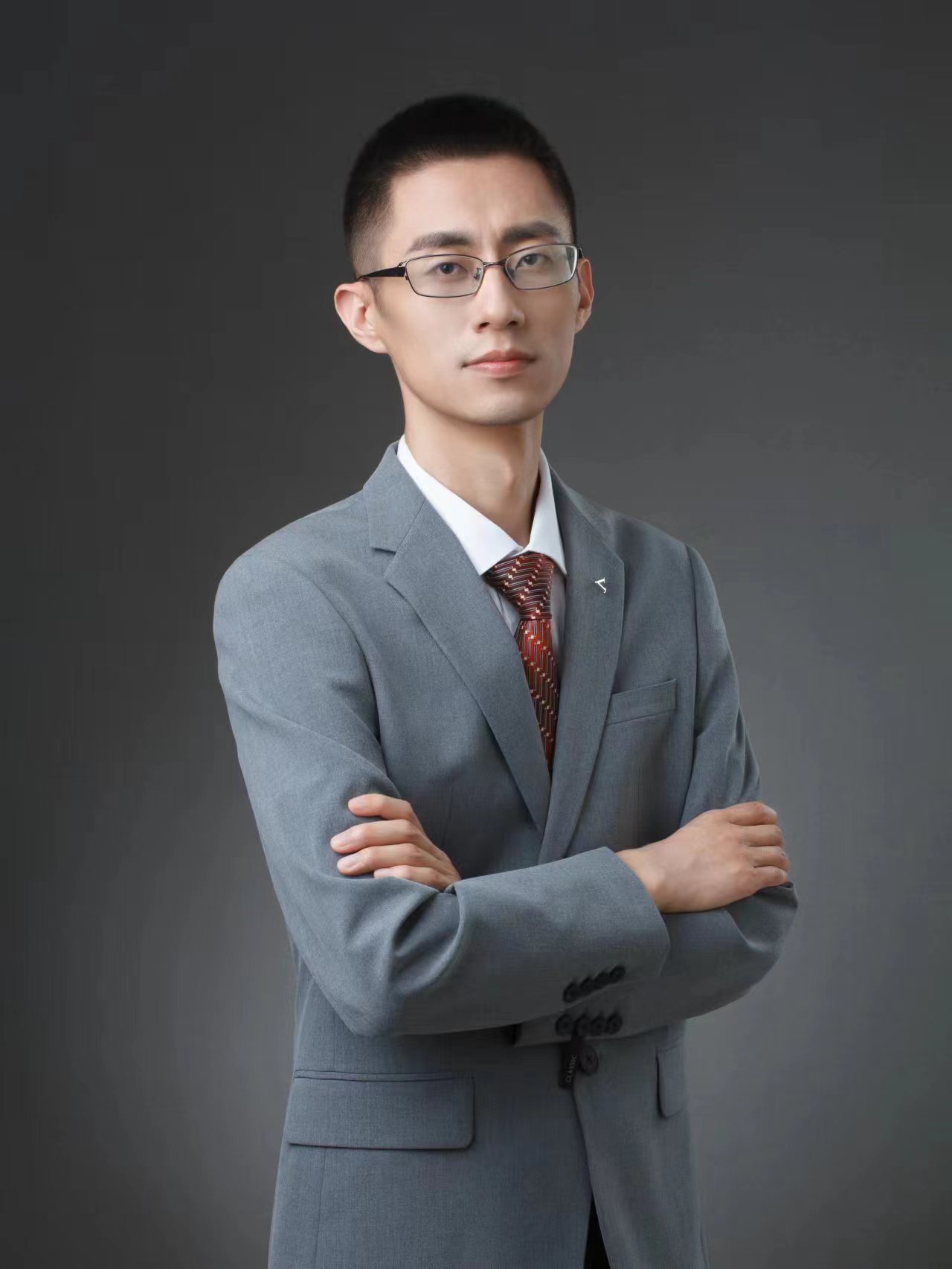}}]{Wei Huang} (IEEE S'18-M'22) received the Ph.D. degree in communication and information system with the School of Electronic Information from Wuhan University, China in 2021. He has published more than 10 SCI/EI research papers.He is now a lecturer and postdoctor at the Faculty of Information Science and Engineering, Ocean University of China.

His current research interests include underwater acoustic tomography, underwater acoustic communication and localization system, and underwater intelligent data processing.
\end{IEEEbiography}

\begin{IEEEbiography}[{\includegraphics[width=1in,height=1.25in,clip,keepaspectratio]{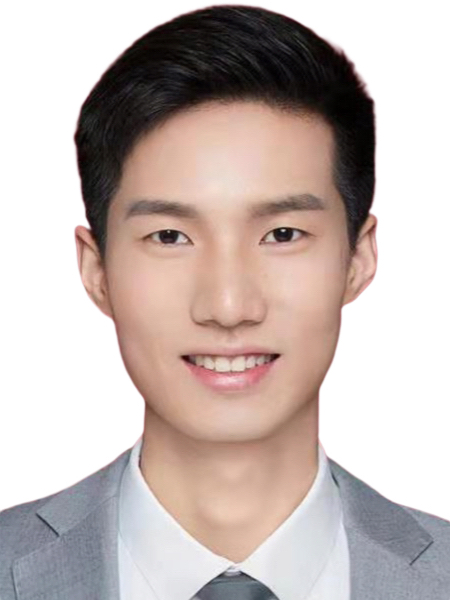}}]{Kaitao Meng} (Member, IEEE) received the B.E. and the Ph.D degrees from the School of Electronic Information, Wuhan University, Wuhan, China, in 2016 and 2021, respectively. From 2021 to 2023, he was a Postdoctoral Researcher in the State Key Laboratory of Internet of Things for Smart City, University of Macau, Macau, China. He is currently a Marie Skłodowska-Curie Actions (MSCA) Postdoctoral Fellow with the Department of Electronic and Electrical Engineering, University College London, U.K. His current research interests mainly include integrated sensing and communication, cooperative sensing, intelligent surfaces, and multi-UAV collaboration. He has served as a TPC Member for many IEEE conferences, such as GLOBECOM and VTC.
\end{IEEEbiography}

\begin{IEEEbiography}[{\includegraphics[width=1in,height=1.25in,clip,keepaspectratio]{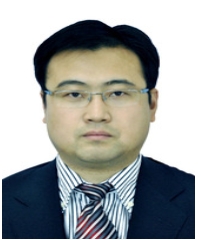}}]{Hao Zhang} (Senior Member, IEEE) received the bachelor's degree in communication engineering from Shanghai Jiao Tong University in 1994, and received the Ph.D. degree in Electronic Engineering from University of Victoria in 2004. He is now a professor and the director of the institute of Ocean Communication at Ocean University of China, and a visting professor at the University of Victoria, Canada. He is one of the Leading Talents in the National "Ten Thousand Talents Plan", he is also one of the Outstanding Talents in the New Century by the Ministry of Education, and a winner of Outstanding Young Scholars in Natural Science in Shandong Province. 
	
He has published more than 150 SCI/EI research papers. His current interests include underwater signal processing, wireless communication, navigation and communication of satellite system, underwater sensor networks.
\end{IEEEbiography}

\begin{IEEEbiography}[{\includegraphics[width=1in,height=1.25in,clip,keepaspectratio]{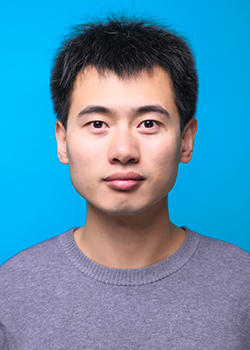}}]{Fan Gao} is now an associate professor with the School of Space Science and Physics, Shandong University. He received his Ph.D. degree in Geodesy and Surveying Engineering from the University of Chinese Academic of Sciences in 2016. His research interests include GNSS-R altimetry, GNSS software-defined receiver, satellite constellation design, precise orbit determination, and underwater navigation.
\end{IEEEbiography}

\begin{IEEEbiography}[{\includegraphics[width=1in,height=1.25in,clip,keepaspectratio]{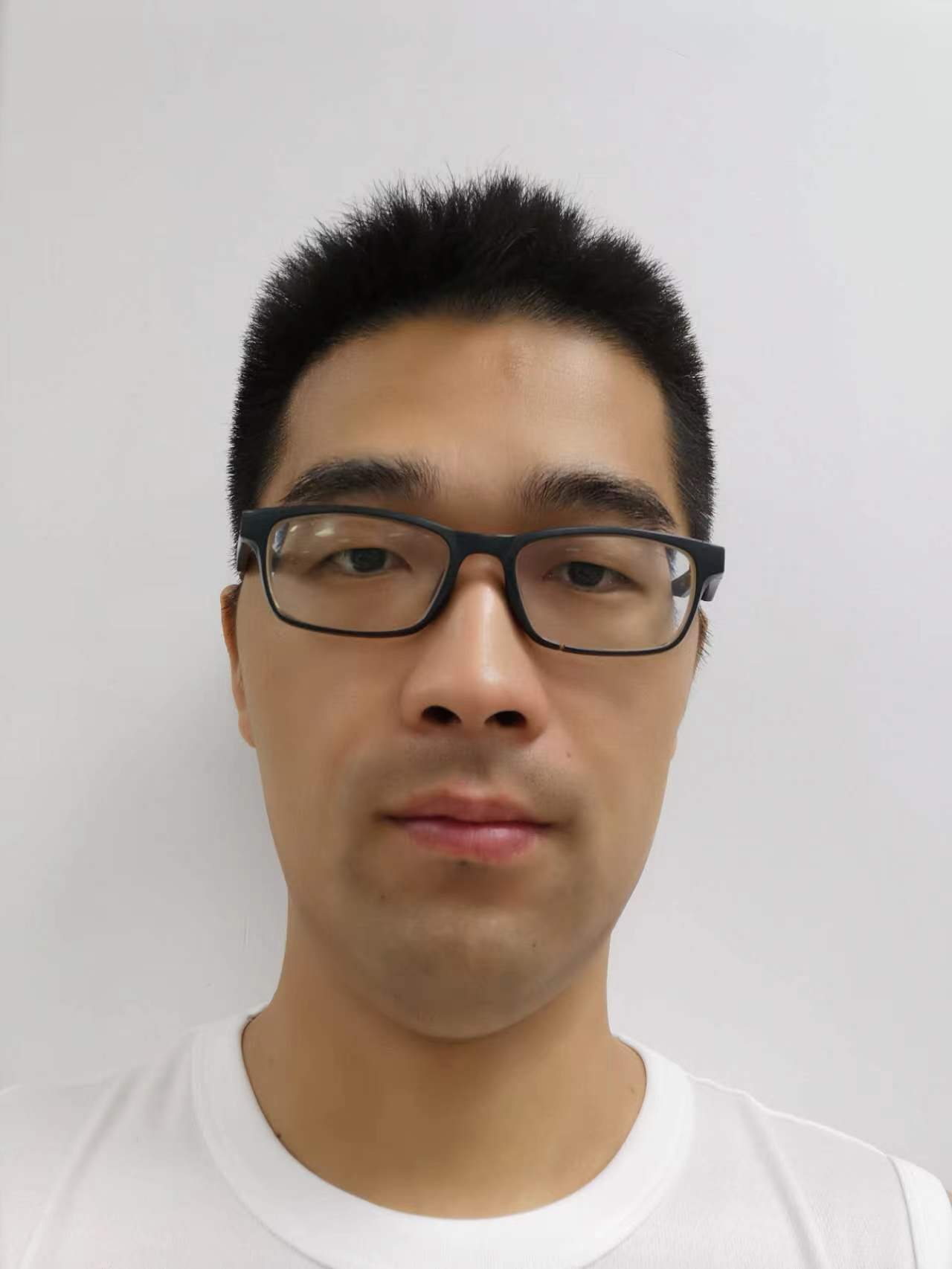}}]{Wenzhou Sun} received the bachelor's degree in Port Channel and Coastal Engineering from Hohai University in 2013, and received the Ph.D. degree in Ocean Surveying and Mapping from Dalian Naval College in 2019. He is now a postdoctor at the State Key Laboratory of Geographic Information. His current research direction is ocean geodesy.
\end{IEEEbiography}

\begin{IEEEbiography}[{\includegraphics[width=1in,height=1.25in,clip,keepaspectratio]{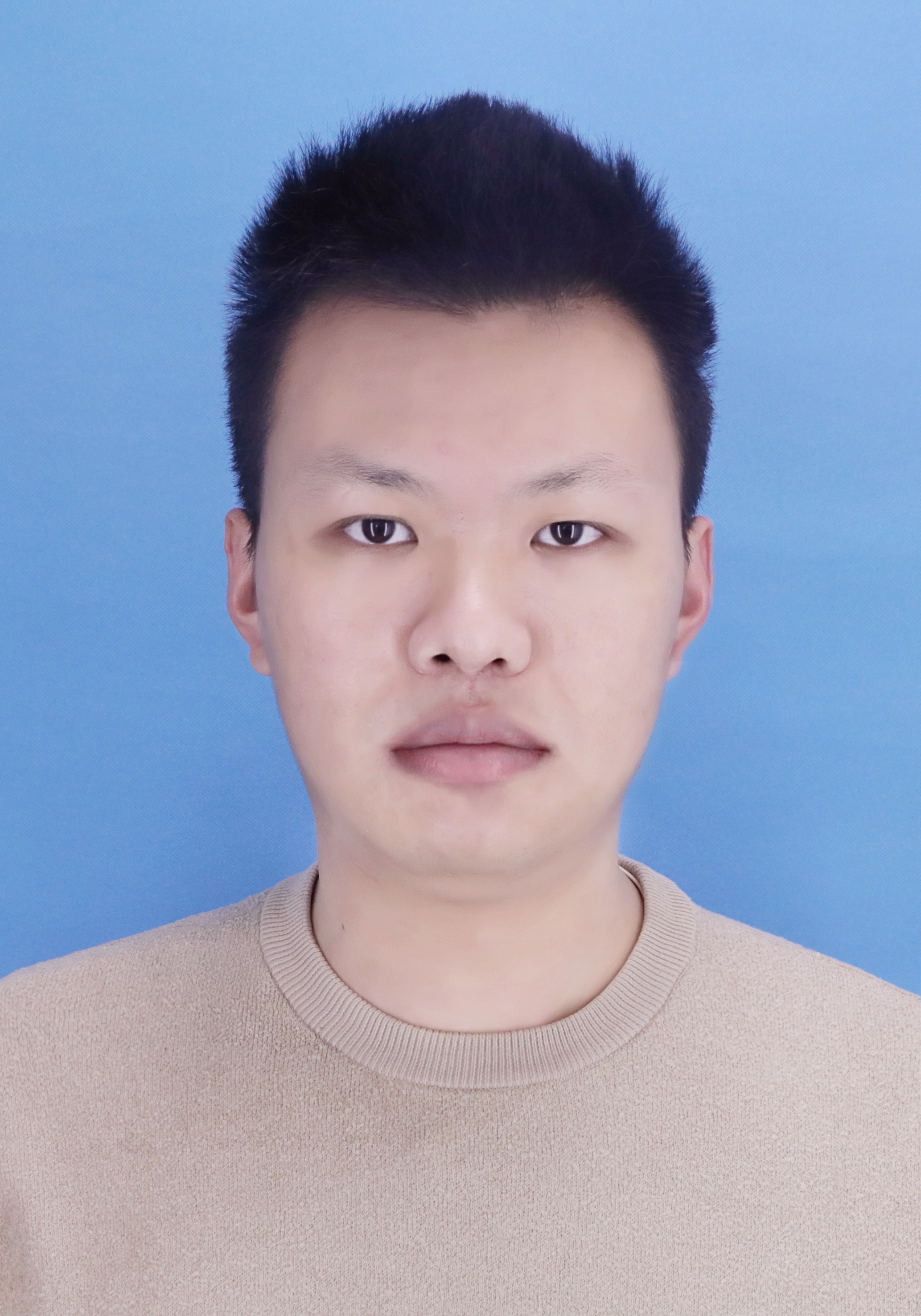}}]{Jianxu Shu} received his M.S. degrees in Geodesy from Chang'an University in 2023. He is now a Ph.D. candidate in Shandong University (Weihai), China. His research interests include underwater localizaiton and navigation.
\end{IEEEbiography}

\begin{IEEEbiography}[{\includegraphics[width=1in,height=1.25in,clip,keepaspectratio]{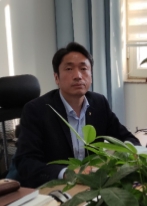}}]{Tianhe Xu} received his Ph.D. and M.S. degrees in Geodesy from the Zhengzhou Institute of Surveying and Mapping of China in 2004 and 2001. He is now a professor at the Institute of Space Sciences in Shandong University, Weihai. He is a chief scientist of National Key Research and Development Program of China. He has published more than 150 SCI/EI papers with citation of more than 3000 times. He has received 2 provincial and ministerial level special awards for scientific and technological progress, 3 first prizes, 7 second prizes, and 1 third prize.	His research interests include satellite navigation, orbit determination, satellite gravity data processing and quality control.
\end{IEEEbiography}

\begin{IEEEbiography}[{\includegraphics[width=1in,height=1.25in,clip,keepaspectratio]{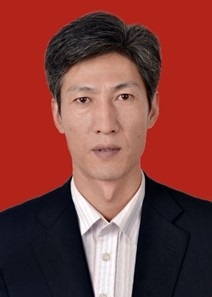}}]{Deshi Li} received his Ph.D. degree in Computer Application Technology from Wuhan University. He has been a visiting scholar with the Network Lab of the University of California at Davis. He is a professor and dean of the Electronic Information School, Wuhan University. Currently, he serves as a member of the Internet of Things Expert Committee, a member of the Education Committee of the Chinese Institute of Electronics, and the Associate Chief Scientist in the area of Space Communication at the Collaborative Innovation Center of Geospatial Technology. 
	
	He has published more than 100 research papers in related areas. His current research interests include wireless communication, the Internet of Things, sensor networks, intelligent systems and SOC design. He is the co-designer of this work, he provides the financial and administrative support, and he is also responsible for the final approval of manuscript.
\end{IEEEbiography}

\vspace{11pt}
\vfill

\end{document}